\documentclass[runningheads]{llncs} 

\usepackage{booktabs}
\usepackage[utf8]{inputenc}
\usepackage{url}
\usepackage{xcolor}
\usepackage{xspace}
\usepackage[all]{xy}
\usepackage{amssymb}
\usepackage{mathtools}
\usepackage{hyperref}
\usepackage{bussproofs}
\usepackage{graphicx}
\usepackage{bm}
\usepackage{extpfeil}
\usepackage{stmaryrd} 
\usepackage{comment}
\usepackage{longtable}

\hypersetup{
    colorlinks=true,
    citecolor=blue,
    linkcolor=blue,
    urlcolor=blue
}

\excludecomment{proceedings}
\includecomment{techreport}

\newcommand{\subst}[2]{\left\{{#2}\middle/{#1}\right\}}
\newcommand{\ssubst}[2]{\left\{{#2}\right\}_{#1}}

\newcommand{\pair}[1]{\langle {#1} \rangle}
\newcommand{\vect}[1]{\overrightarrow{#1}}

\newcommand{\cE}[0]{\mathcal{E}}
\newcommand{\cF}[0]{\mathcal{F}}
\newcommand{\cI}[0]{\mathcal{I}}

\newcommand{\cQ}[0]{\mathcal{Q}}
\newcommand{\cR}[0]{\mathcal{R}}
\newcommand{\cS}[0]{\mathcal{S}}
\newcommand{\cT}[0]{\mathcal{T}}

\newcommand{\derives}[0]{\mathrel{\triangleright}}
\newcommand{\eeq}[1]{\xleftrightarrow{#1}}
\newcommand{\red}[1]{\xrightarrow{#1}}
\newcommand{\rred}[1]{\xtwoheadrightarrow{#1}}
\newcommand{\bred}[0]{\red{\;\; \beta \;\;}}
\newcommand{\bbred}[0]{\red{\;\; \bar{\beta} \;\;}}
\newcommand{\sred}[0]{\red{\;\; \sigma \;\;}}

\newcommand{\tred}[0]{\red{\;\; \tau \;\;}}

\newcommand{\stred}[0]{\red{\; \sigma\tau \;}}
\newcommand{\steeq}[0]{\eeq{\;\; \sigma\tau \;\;}}
\newcommand{\Bred}[0]{\red{\mbox{Beta}}}
\newcommand{\brred}[0]{\rred{\beta}}

\newcommand{\bbrred}[0]{\rred{\bar{\beta}}}
\newcommand{\srred}[0]{\rred{\sigma}}
\newcommand{\trred}[0]{\rred{\tau}}
\newcommand{\strred}[0]{\rred{\sigma\tau}}
\newcommand{\Brred}[0]{\rred{\mbox{Beta}}}
\newcommand{\CAUmsred}[0]{\red{\CAUms}}
\newcommand{\CAUmsrred}[0]{\rred{\CAUms}}
\newcommand{\CAUmred}[0]{\red{\CAUm}}
\newcommand{\CAUmrred}[0]{\rred{\CAUm}}

\newcommand{\qrefl}[0]{\mathbf{r}}

\newcommand{\qtrans}[0]{\mathbf{t}}
\newcommand{\qapp}[0]{\mathbf{app}}
\newcommand{\qlam}[0]{\mathbf{lam}}
\newcommand{\qlet}[0]{\mathbf{let_!}}
\newcommand{\qti}[0]{\mathbf{ti}}
\newcommand{\qtr}[0]{\mathbf{tb}}

\newcommand{\qba}[0]{{\bm{\beta}}}
\newcommand{\qbb}[0]{{\bm{\beta_!}}}

\newcommand{\AU}[1]{\left\llbracket{#1}\right\rrbracket}

\newcommand{\bang}{\mathord{!}}

\def\orelse{\mathbin{|}}
\newcommand{\erase}[1]{\left\lfloor {#1} \right\rfloor}
\newcommand{\trailify}[1]{\left\lceil {#1} \right\rceil}
\newcommand{\Erase}[1]{\llfloor {#1} \rrfloor}
\newcommand{\Trailify}[1]{\llceil {#1} \rrceil}
\newcommand{\focus}[1]{\left\| {#1} \right\|}
\def\id{\pair{}}
\def\comp{\mathbin{\circ}}
\def\lift{\mathbin{\uparrow}}
\newcommand{\CAU}{{\ensuremath{\textbf{CAU}}}\xspace}
\def\CAUm{{\ensuremath{\textbf{CAU}^-}}\xspace}
\def\CAUms{{\ensuremath{\textbf{CAU}^-_\sigma}}\xspace}
\def\tinsp{\iota}
\def\Hole{\blacksquare}
\def\dec{\underline}
\newcommand{\tmdec}[1]{\cT({#1})}

\newcommand{\jtm}[1]{{#1}~\mathbf{tm}}
\newcommand{\jctx}[1]{{#1}~\mathbf{ctx}}

\DeclareMathOperator{\tlet}{let_!}
\DeclareMathOperator{\tlift}{lift}
\def\kwplus{\mathit{plus}}

\def\kwfact{\mathit{fact}}
\def\kwsum{\mathit{sum}}

\def\rVS{\noalign{\vskip 2ex plus.4ex minus.4ex}}

\input{widebar}

\newcommand{\pure}[1]{\widebar{#1}}

\begin{document}

\title{Explicit Auditing}

\author{
	Wilmer Ricciotti
	\and
	James Cheney
}
\authorrunning{W. Ricciotti and J. Cheney}
\institute{
	LFCS, University of Edinburgh
	\\
	\email{research@wilmer-ricciotti.net}
	\\
	\email{jcheney@inf.ed.ac.uk}
}

\maketitle

\begin{abstract}
The Calculus of Audited Units (CAU) is a typed lambda calculus
resulting from a computational interpretation of Artemov's
Justification Logic under the Curry-Howard isomorphism; it extends the
simply typed lambda calculus by providing \emph{audited types},
inhabited by expressions carrying a \emph{trail} of their past
computation history. Unlike most other auditing techniques, CAU allows
the inspection of trails at runtime as a first-class operation, with
applications in security, debugging, and transparency of scientific computation.

An efficient implementation of CAU is challenging: not only do the
sizes of trails grow rapidly, but they also need to be normalized after
every beta reduction. In this paper, we study how to reduce terms more
efficiently in an untyped variant of CAU by means of explicit
substitutions and explicit auditing operations, finally deriving a call-by-value abstract machine. 
\end{abstract}

\section{Introduction}
Transparency is an increasing concern in computer systems: for complex
systems, whose desired behavior may be difficult to formally specify,
auditing is an important complement to traditional techniques for
verification and static analysis for
security~\cite{Abadi2003,amir-mohamedian2016,Banerjee2005,vaughan08csf,jia08icfp,Garg2011},
program slicing~\cite{Perera2012,Ricciotti2017b}, and 
provenance~\cite{Moreau2010,Ricciotti2017c}. However, formal foundations of auditing as a
programming language primitive are not yet well-established: most
approaches view auditing as an extra-linguistic operation, rather than
a first-class construct.  Recently, however, Bavera and
Bonelli~\cite{Bavera2015} introduced 
a calculus in which
recording and analyzing audit trails are first-class operations.  They
proposed a $\lambda$-calculus based on a Curry-Howard correspondence
with Justification
Logic~\cite{Artemov2001,Artemov2007,Artemov2008,Artemov2008b} called \emph{calculus of audited units}, or
\CAU. In recent work, we developed a simplified form of \CAU and proved
strong normalization~\cite{Ricciotti2017a}.

The type system of \CAU is based on modal
logic, following Pfenning and Davies~\cite{Pfenning2001}: it provides a type
$\AU{s}{A}$ of audited units, where $s$ is ``evidence'',
or the expression that was evaluated to produce the result of type
$A$.  Expressions of this type $\bang_qM$ contain a value of type $A$ along
with a ``trail'' $q$ explaining how $M$ was obtained by evaluating
$s$.  Trails are essentially (skeletons of) proofs of reduction of
terms, which can be \emph{inspected} by structural recursion using a special language construct.

To date, most work on foundations of auditing has focused on design,
semantics, and correctness properties, and relatively little attention
has been paid to efficient execution, while most work on auditing systems has
neglected these foundational aspects.  Some work on tracing and
slicing has investigated the use of ``lazy''
tracing~\cite{Perera2012}; however, to the best of our knowledge there
is no prior work on how to efficiently evaluate a language such as
\CAU in which auditing is a built-in operation.
This is the problem studied in this paper.

 A na\"ive approach to implementing the semantics of \CAU as given by
Bavera and Bonelli runs immediately into the following problem: a \CAU
reduction first performs a \emph{principal contraction} (e.g. beta
reduction), which typically introduces a local trail annotation
describing the reduction, that can block further beta-reductions. The
local trail annotations are then moved up to the nearest enclosing
audited unit constructor using one or more \emph{permutation
  reductions}.  For example:
\begin{align*}
\bang_q \cF[(\lambda x.M)~N] 
& 
\bred \bang_q \cF [\qba \derives M\subst{x}{N}]
\\
&
\trred\bang_{\qtrans(q,\cQ[\qba])} \cF [M\subst{x}{N}]
\end{align*}
where $\cF[]$ is a bang-free evaluation context and $\cQ[\qba]$ is a
subtrail that indicates where in context $\cF$ the $\qba$-step was
performed.  As the size of the term being executed (and distance
between an audited unit constructor and the redexes) grows, this
evaluation strategy slows down quadratically in the worst case;
eagerly materializing the traces likewise imposes additional storage
cost.

While some computational overhead seems inevitable to accommodate
auditing, both of these costs can in principle be mitigated.  Trail
permutations are computationally expensive and can often be delayed
without any impact on the final outcome. Pushing trails to the closest
outer bang does not serve any real purpose: it would be more efficient
to keep the trail where it was created and perform normalization only
if and when the trail must be inspected (and this operation does not
even actually require an actual pushout of trails, because we can
reuse term structure to compute the trail structure on-the-fly).

This situation has a well-studied analogue: in the $\lambda$-calculus,
it is not necessarily efficient to eagerly perform all substitutions
as soon as a $\beta$-reduction happens.  Instead, calculi of
\emph{explicit substitutions} such as Abadi et al.'s
$\lambda\sigma$~\cite{lambdasigma} have been developed in which
substitutions are explicitly tracked and rewritten.  Explicit
substitution calculi have been studied extensively as a bridge between
the declarative rewriting rules of $\lambda$-calculi and efficient
implementations.  Inspired by this work, we hypothesize that calculi
with auditing can be implemented more efficiently by delaying the
operations of trail extraction and erasure, using explicit symbolic
representations for these operations instead of performing them
eagerly. 
 
Particular care must be placed in making sure that the trails we produce still
correctly describe the order in which operations were actually
performed (e.g. respecting call-by-name or call-by-value reduction):
when we perform a principal contraction, pre-existing trail
annotations must be recorded as history that happened \emph{before}
the contraction, and not after. In the original eager reduction style,
this is trivial because we never contract terms containing trails;
however, we will show that, thanks to the explicit trail operations, 
correctness can be achieved even when adopting a lazy normalization of trails.
\begin{techreport}
Accordingly, we will introduce explicit terms for \emph{delayed
trail erasure} $\erase{M}$ and \emph{delayed trail extraction} 
$\trailify{M}$.  
We can use these features to decrease the cost of normalization: for instance, the $\beta$-reduction above can be replaced by a rule with delayed treatment of substitution and trails, denoted by $\mathrm{Beta}$:
\begin{align*}
&
\cF[(\lambda.M)~N] \Bred
\\
& 
\quad \bang_q \cF [\qtrans(\qapp(\qlam(\trailify{M}),\trailify{N}),\qba) \derives \erase{M}[\erase{N}]]
\end{align*}
Here, we use de Bruijn notation~\cite{debruijn72} (as in
$\lambda\sigma$, and anticipating Sections~\ref{naive} and
\ref{CAUms}), and write $M[N]$ for the explicit substitution of $N$
for the outermost bound variable of $\lambda.M$.  The trail
constructor $\qtrans$ stands for transitive composition of trails,
while $\qapp$ and $\qlam$ are congruence rules on trails, so the trail
$\qtrans(\qapp(\qlam(\trailify{M}),\trailify{N}),\qba) $ says that the
redex's trail is constructed by extracting the latent trail
information from $M$ and $N$, combining it appropriately, and then
performing a $\qba$ step.  The usual contractum itself is obtained by
substituting the erased argument $\erase{N}$ into the erased function
body $\erase{M}$.  Although this may look a bit more verbose than the
earlier beta-reduction, the additional work done to create
the trail $\qtrans(\qapp(\qlam(\trailify{M}),\trailify{N}),\qba)$ is
all work that would have been done anyway using the eager system,
while the use of lazy trail-extraction and trail-erasure operations
gives us many more ways to do the remaining work
efficiently --- for example, if the trail is never subsequently used,
we can just discard it without doing any more work.
\end{techreport}

\paragraph*{Contributions}
We study an extension of Abadi et al.'s 
calculus
$\lambda\sigma$~\cite{lambdasigma} with explicit auditing operations.  We consider a
simplified, untyped variant \CAUm of the Calculus of Audited Units
(Section~\ref{CAUm}); this simplifies our presentation because type
information is not needed during execution.  We revisit
$\lambda\sigma$ in Section~\ref{naive}, extend it to include auditing
and trail inspection features, and discuss problems with this initial,
na\"ive approach.  We address these problems by developing a new calculus \CAUms with explicit versions of the ``trail extraction'' and ``trail erasure'' operations (Section~\ref{CAUms}), and we show that it correctly refines \CAUm (subject to an
obvious translation). In Section~\ref{SECD}, we build on \CAUms to define an abstract machine for audited computation and prove its correctness.
\begin{proceedings}
Some proofs have been omitted due to space constraints and are included in the extended version of this paper.\footnote{Available online at \url{http://www.yyy.zzz}}
\end{proceedings}
\begin{techreport}
Details of the proofs are included in the appendix.
\end{techreport}


\section{The Untyped Calculus of Audited Units}\label{CAUm}

The language \CAUm presented here is an untyped version of the calculi
$\lambda^h$~\cite{Bavera2015} and Ricciotti and
Cheney's $\lambda^{hc}$~\cite{Ricciotti2017a} obtained by erasing all typing
information and a few other related technicalities: this will allow us
to address all the interesting issues related to the reduction of \CAU
terms, but with a much less pedantic syntax. To help us explain the
details of the calculus, we adapt some examples from our previous
paper~\cite{Ricciotti2017a}; other examples are described by Bavera
and Bonelli~\cite{Bavera2015}.

Unlike the typed variant of the calculus, we only need one sort of
variables, denoted by the letters $x, y, z \ldots$. The syntax of \CAUm is as follows:

\begin{tabular}{llcl}
\textbf{Terms} & $M,N$ & $::=$ & $x \orelse \lambda x.M \orelse M~N \orelse \tlet(x := M,N) \orelse \bang_q M \orelse q \triangleright M \orelse \tinsp(\vartheta)$ \\
\textbf{Trails} & $q,q'$ & $::=$ & $\qrefl \orelse \qtrans(q,q') \orelse \qba \orelse \qbb \orelse \qti \orelse \qlam(q) \orelse \qapp(q,q') \orelse \qlet(q,q') \orelse \qtr(\zeta)$ 
\end{tabular}

\CAUm extends the pure lambda calculus with \emph{audited units}
$\bang_q M$ (colloquially, ``bang $M$''), whose purpose is to decorate the term $M$ with a log $q$ of its computation history, called \emph{trail} in our terminology: when $M$ evolves as a result of computation, $q$ will be updated by adding information about the reduction rules that have been applied. The form $\bang_q M$ is in general not intended for use in source programs: instead, we will write $\bang~M$ for $\bang_\qrefl M$, where $\qrefl$ represents the empty execution history (\emph{reflexivity} trail).

Audited units can then be employed in larger terms by means of the ``let-bang'' operator, which unpacks an audited unit and thus allows us to access its contents. The variable declared by a $\tlet$ is bound in its second argument: in essence $\tlet(x := \bang_q M,N)$ will reduce to $N$, where free occurrences of $x$ have been replaced by $M$; the trail $q$ will not be discarded, but will be used to produce a new trail explaining this reduction.

The expression form $q \derives M$ is an auxiliary, intermediate annotation of $M$ with partial history information which is produced during execution and will eventually stored in the closest surrounding bang.

\begin{example}\label{ex:letbang}
In \CAUm we can express history-carrying terms explicitly: for instance, if we use $\bar{n}$ to denote the Church encoding of a natural number $n$, and $\kwplus$ or $\kwfact$ for lambda terms computing addition and factorial on said representation, we can write audited units like
\[
\bang_q \bar{2} \qquad \bang_{q'} \bar{6}
\]
where $q$ is a trail representing the history of $\bar{2}$ i.e., for instance, a witness for the computation that produced $\bar{2}$ by reducing $\kwplus~\bar{1}~\bar{1}$; likewise, $q'$ might describe how computing $\kwfact~\bar{3}$ produced $\bar{6}$.
Supposing we wish to add these two numbers together, at the same time retaining their history, we will use the $\tlet$ construct to look inside them:
\[
\tlet(x := \bang_q \bar{2}, \tlet(y := \bang_{q'} \bar{6}, \kwplus~x~y)) \rred{} q'' \derives \bar{8}
\]
where the final trail $q''$ is produced by composing $q$ and $q'$; if this reduction happens inside an external bang, $q''$ will eventually be captured by it. 
\end{example}

Trails, representing sequences of reduction steps, encode the
(possibly partial) computation history of a given subterm. The main
building blocks of trails are $\qba$ (representing standard beta
reduction), $\qbb$ (contraction of a let-bang redex) and $\qti$
(denoting the execution of a trail inspection). For every class of
terms we have a corresponding congruence trail
($\qlam, \qapp, \qlet, \qtr$, the last of which is associated with
trail inspections), with the only exception of bangs, which do not
need a congruence rule because they capture all the computation
happening inside them. The syntax of trails is completed by
reflexivity $\qrefl$ (representing a null computation history, i.e. a
term that has not reduced yet) and transitivity $\qtrans$
(i.e. sequential composition of execution steps).  As discussed by our
earlier paper~\cite{Ricciotti2017a}, we omit Bavera and Bonelli's symmetry trail
form.

\begin{example}\label{ex:applambda}
We build a pair of natural numbers using Church's encoding: 
\begin{align*}
\bang~((\lambda x,y,p.p~x~y)~2)~6 
& 
\red{} \bang_{\qtrans(\qrefl,\qapp(\qba,\qrefl))}~(\lambda y,p.p~2~y)~6
\\
& \red{} \bang_{\qtrans(\qtrans(\qrefl,\qapp(\qba,\qrefl)),\qba)}~\lambda p.p~2~6
\end{align*}
The trail for the first computation step is obtained by transitivity (trail constructor $\qtrans$) from the original trivial trail ($\qrefl$, i.e. reflexivity) composed with $\qba$, which describes the reduction of the applied lambda: this subtrail is wrapped in a congruence $\qapp$ because the reduction takes place deep inside the left-hand subterm of an application (the other argument of $\qapp$ is reflexivity, because no reduction takes place in the right-hand subterm).

The second beta-reduction happens at the top level and is thus not wrapped in a congruence. It is combined with the previous trail by means of transitivity.
\end{example}

The last term form $\tinsp(\vartheta)$, called \emph{trail inspection}, will perform primitive recursion on the computation history of the current audited unit. The metavariables $\vartheta$ and $\zeta$ associated with trail inspections are \emph{trail replacements}, i.e. maps associating to each possible trail constructor, respectively, a term or a trail:
\begin{align*}
\vartheta ::= 
& 
\{ M_1/\qrefl, M_2/\qtrans, M_3/\qba, M_4/\qbb,M_5/\qti, M_6/\qlam, M_7/\qapp,M_8/\qlet, M_9/\qtr \} 
\\
\zeta ::= 
& 
\{ q_1/\qrefl, q_2/\qtrans,q_3/\qba, q_4/\qbb,q_5/\qti, q_6/\qlam, q_7/\qapp, q_8/\qlet, q_9/\qtr \}
\end{align*}
When the trail constructors are irrelevant for a certain $\vartheta$
or $\zeta$, we will omit them, using the notations $\{
\vect{M} \}$ or $ \{ \vect{q} \}$.
These constructs represent (or describe) the nine cases of a
structural recursion operator over trails, which we write as
$q\vartheta$.

\begin{definition}
The operation $q\vartheta$, which produces a term by structural recursion on $q$ applying the inspection branches $\vartheta$, is defined as follows:
\[
\begin{array}{c}
\qrefl\vartheta \triangleq \vartheta(\qrefl) 
\qquad
\qtrans(q,q')\vartheta \triangleq
  \vartheta(\qtrans)~(q\vartheta)~(q'\vartheta)
\qquad
\qba\vartheta \triangleq \vartheta(\qba) 
\qquad
\qbb\vartheta \triangleq \vartheta(\qbb)

\\

\qti\vartheta \triangleq \vartheta(\qti) 
\qquad
\qlam(q)\vartheta \triangleq \vartheta(\qlam)~(q\vartheta)
\qquad
\qtr(\{\vect{q}\})\vartheta \triangleq \vartheta(\qtr)~\vect{(q\vartheta)}

\\

\qapp(q,q')\vartheta \triangleq \vartheta(\qapp)~(q\vartheta)~(q'\vartheta)
\qquad
\qlet(q,q')\vartheta \triangleq \vartheta(\qlet)~(q\vartheta)~(q'\vartheta) 
\\
\end{array}
\]
where the sequence $\vect{(q\vartheta)}$ is obtained from $\vect{q}$ by pointwise recursion.
\end{definition}
\begin{example}\label{ex:trplus}
Trail inspection can be used to count all of the contraction steps in the history of an audited unit, by means of the following trail replacement:
\begin{align*}
\vartheta_+  ::= 
&
\{ \bar{0}/\qrefl, \kwplus/\qtrans, \bar{1}/\qba, \bar{1}/\qbb, \bar{1}/\qti, \lambda x.x/\qlam, \kwplus/\qapp,\kwplus/\qlet, \kwsum/\qtr \}
\end{align*}
where $\kwsum$ is a variant of $\kwplus$ taking nine arguments, as required by the arity of $\qtr$. For example, we can count the contractions in $q = \qtrans(\qlet(\qba,\qrefl),\qbb)$ as:
\[
q\vartheta_+ = \kwplus~(\kwplus~\bar{1}~\bar{0})~\bar{1}
\]
\end{example}

\subsection{Reduction}
Reduction in \CAUm 
includes rules to contract the usual beta redexes (applied lambda abstractions), ``beta-bang'' redexes, which unpack the bang term
appearing as the definiens of a $\tlet$, and trail inspections. These rules, which we call \emph{principal contractions}, are defined as follows:
\begin{gather*}
(\lambda x.M)~N  \bred \qba \derives M\subst{x}{N} 
\qquad
\tlet(x := \bang_q M,N) \bred \qbb \derives N\subst{x}{q \derives M} 
\\
\bang_q \cF[\tinsp(\vartheta)] \bred \bang_q \cF[\qti \derives q\vartheta]
\end{gather*}
Substitution $M\subst{x}{N}$ is defined in the traditional way, 
avoiding variable capture. 
The first contraction is familiar, except for the fact that the reduct $M\subst{x}{N}$ has been annotated with a $\qba$ trail. The second one deals with unpacking a bang: from $\bang_q M$ we obtain $q \derives M$, which is then substituted for $x$ in the target term $N$; the resulting term is annotated with a $\qbb$ trail.
The third contraction defines the result of a trail inspection $\tinsp(\vartheta)$. Trail inspection will be contracted by capturing the current history, as stored in the nearest enclosing bang, and performing structural recursion on it according to the branches defined by $\vartheta$. 
The concept of ``nearest enclosing bang'' is made formal by contexts $\cF$ in which the hole cannot appear inside a bang (or \emph{bang-free} contexts, for short):
\begin{align*}
\cF & ::= 
\blacksquare \orelse \lambda x.\cF \orelse \cF~M \orelse M~\cF \orelse \tlet(\cF,M) \orelse \tlet(M,\cF) \orelse q \derives \cF \orelse \tinsp(\{ \vect{M},\cF,\vect{N} \})
\end{align*}
The definition of the principal contractions is completed, as usual, by a
contextual closure rule stating that they can appear in any context $\cE$:
\begin{align*}
\cE & ::= 
\blacksquare \orelse \lambda x.\cE \orelse \cE~M \orelse M~\cE
  \orelse \tlet(\cE,M) \orelse \tlet(M,\cE) \orelse \bang_q \cE
  \orelse q \derives \cE \orelse \tinsp(\{ \vect{M},\cE,\vect{N} \})
\end{align*}
\begin{prooftree}
\AxiomC{$M \bred N$}
\UnaryInfC{$\cE[M] \bred \cE[N]$}
\end{prooftree}
The principal contractions introduce local trail subterms
$q' \derives M$, which can block other reductions. Furthermore, the rule for trail inspection assumes that the $q$ annotating the enclosing bang really is a complete log of the history of the audited unit; but at the same time, it violates this invariant, because the $\qti$ trail created after the contraction is not merged with the original history $q$.

For these reasons, we only want to perform principal contractions on terms not containing local trails: after each principal contraction, we apply the following rewrite rules, called \emph{permutation reductions}, to ensure that the local trail is moved to the nearest enclosing bang:
\[
\begin{array}{c}
\begin{array}{rl@{\hspace{2em}}rl}
\qrefl \derives M & \tred M
&
q \derives (q' \derives M) 
& \tred \qtrans(q, q') \derives M
\\
\bang_q (q' \triangleright M) & \tred \bang_{\qtrans(q,q')} M
&
\lambda x. (q \derives M) & \tred \qlam(q) \derives \lambda x.M 
\\
(q \derives M)~N & \tred \qapp(q,\qrefl) \derives M~N 
&
M~(q \derives N) & \tred \qapp(\qrefl,q) \derives M~N
\end{array}
\\
\begin{array}{rl}
\tlet(x := q \triangleright M, N) & \tred \qlet(q,\qrefl) \triangleright \tlet(x := M,N)
\\
\tlet(x := M, q \triangleright N) & \tred \qlet(\qrefl,q) \triangleright \tlet(x := M,N)
\\
\tinsp(\{M_1,\ldots,q \derives M_i,\ldots,M_9 \}) & \tred \qtr(\{\qrefl,\ldots,q,\ldots,\qrefl\}) \derives \tinsp(\{M_1,\ldots,M_9\})
\end{array}
\end{array}
\]
Moreover, the following rules are added to the $\tred$ relation to ensure confluence:
\[
\begin{array}{c}
\begin{array}{rl@{\hspace{2em}}rl@{\hspace{2em}}rl}
\qtrans(q,\qrefl) & \tred q
&
\qtrans(\qrefl,q) & \tred q 
&
\qtr(\{\vect{\qrefl}\}) & \tred \qrefl
\\
\qapp(\qrefl,\qrefl) & \tred \qrefl
& 
\qlam(\qrefl) & \tred \qrefl
&
\qlet(\qrefl,\qrefl) & \tred \qrefl
\end{array}
\\
\begin{array}{rl}
 \qtrans(\qtrans(q_1,q_2),q_3) & \tred \qtrans(q_1,\qtrans(q_2,q_3))
 \\
\qtrans(\qlam(q),\qlam(q')) & \tred \qlam(\qtrans(q,q')) 
\\
\qtrans(\qlam(q_1),\qtrans(\qlam(q_1'),q)) & \tred \qtrans(\qlam(\qtrans(q_1,q'_1)),q) 
\\
\qtrans(\qapp(q_1,q_2),\qapp(q'_1,q'_2)) & \tred \qapp(\qtrans(q_1,q'_1),\qtrans(q_2,q'_2)) 
\\
\qtrans(\qapp(q_1,q_2),\qtrans(\qapp(q'_1,q'_2)),q) & \tred \qtrans(\qapp(\qtrans(q_1,q'_1),\qtrans(q_2,q'_2)),q) 
\\
\qtrans(\qlet(q_1,q_2),\qlet(q'_1,q'_2)) & \tred \qlet(\qtrans(q_1,q'_1),\qtrans(q_2,q'_2)) 
\\
\qtrans(\qlet(q_1,q_2),\qtrans(\qlet(q'_1,q'_2)),q) & \tred \qtrans(\qlet(\qtrans(q_1,q'_1),\qtrans(q_2,q'_2)),q)
\\
\qtrans(\qtr(\vect{q_1}),\qtr(\vect{q_2})) & \tred \qtr(\vect{\qtrans(q_1,q_2)}) 
\\
\qtrans(\qtr(\vect{q_1}),\qtrans(\qtr(\vect{q_2}),q)) & \tred \qtrans(\qtr(\vect{\qtrans(q_1,q_2)}),q)
\end{array}
\end{array}
\]
As usual, $\tred$ is completed by a contextual closure rule. We prove

\begin{lemma}[\cite{Bavera2015}]
$\tred$ is terminating and confluent.
\end{lemma}
When a binary relation $\red{\cR}$ on terms is terminating and confluent, we will write $\cR(M)$ for the unique $\cR$-normal form of $M$.
Since principal contractions must be performed on $\tau$-normal terms, it is convenient to merge contraction and $\tau$-normalization in a single operation, which we will denote by $\CAUmred$:
\begin{prooftree}
\AxiomC{$M \bred N$}
\UnaryInfC{$M \CAUmred \tau(N)$}
\end{prooftree}

\begin{example}
We take again the term from Example~\ref{ex:letbang} and reduce the outer $\tlet$ as follows:
\[
\begin{array}{ll}
\multicolumn{2}{l}{
\bang~\tlet(x := \bang_q 2, \tlet(y := \bang_{q'} 6, \kwplus~x~y))
}
\\
\bred 
&
\bang~(\qbb \derives \tlet(y := \bang_{q'} 6, \kwplus~(q \derives 2)~y))
\\
\trred
&
\bang_{\qtrans(\qbb, \qlet(\qrefl,\qapp(\qapp(\qrefl,q),\qrefl)))}
\tlet(y := \bang_{q'} 6, \kwplus~2~y)
\end{array}
\]
This $\tlet$-reduction substitutes $q \derives 2$ for $x$; a $\qbb$ trail is produced immediately inside the bang, in the same position as the redex. Then, we $\tau$-normalize the resulting term, which results in the two trails being combined and used to annotate the enclosing bang.
\end{example}

\section{Na\"ive explicit substitutions}\label{naive}
We seek to adapt the existing abstract machines for the efficient normalization of lambda terms to \CAUm. Generally speaking, most abstract machines act on nameless terms, using de~Bruijn's indices~\cite{debruijn72}, thus avoiding the need to perform renaming to avoid variable capture when substituting a term into another. 

Moreover, since a substitution $M\subst{x}{N}$ requires to scan the whole term $M$ and is thus \emph{not} a constant time operation, it is usually not executed immediately in an eager way. The abstract machine actually manipulates \emph{closures}, or pairs of a term $M$ and an environment $s$ declaring lazy substitutions for each of the free variables in $M$: this allows $s$ to be applied in an incremental way, while scanning the term $M$ in search for a redex.
In the $\lambda\sigma$-calculus of Abadi et al. \cite{lambdasigma}, lazy substitutions and closures are manipulated explicitly, providing an elegant bridge between the classical $\lambda$-calculus and its concrete implementation in abstract machines. Their calculus expresses beta reduction as the rule
\[
(\lambda.M)~N \longrightarrow M[N]
\]
where $\lambda.M$ is a nameless abstraction \emph{\`a la de~Bruijn}, and $[N]$ is a (suspended) \emph{explicit substitution} mapping the variable corresponding to the first dangling index in $M$ to $N$, and all the other variables to themselves. Terms in the form $M[s]$, representing closures, are syntactically part of $\lambda\sigma$, as opposed to substitutions $M\subst{x}{N}$, which are meta-operations that \emph{compute} a term.
In this section we formulate a first attempt at adding explicit substitutions to \CAUm. We will not prove any formal result for the moment, as our purpose is to elicit the difficulties of such a task. An immediate adaptation of $\lambda\sigma$-like explicit substitutions yields the following syntax:
\begin{center}
\begin{tabular}{lll}
\textbf{Terms} & $M,N$ & $::= 1 \orelse \lambda.M \orelse M~N \orelse \tlet(M,N) \orelse \bang_q M \orelse q \triangleright M \orelse \tinsp(\vartheta) \orelse M[s]$ 
\\
\textbf{Substitutions} & $s,t$ & $::= \id \orelse \lift \orelse s \circ t \orelse M \cdot s$
\end{tabular}
\end{center}
where $1$ is the first de~Bruijn index, the nameless $\lambda$ binds the first free index of its argument, and similarly the nameless $\tlet$ binds the first free index of its second argument. Substitutions include the identity (or empty) substitution $\id$, lift $\lift$ (which reinterprets all free indices $n$ as their successor $n+1$), the composition $s \circ t$ (equivalent to the sequencing of $s$ and $t$) and finally $M\cdot s$ (indicating a substitution that will replace the first free index with $M$, and other indices $n$ with their predecessor $n-1$ under substitution $s$). Trails are unchanged.

We write $M[N_1 \cdots N_k]$ as syntactic sugar for $M[N_1 \cdots N_k \cdot \id]$. Then, \CAUm reductions can be expressed as follows:
\begin{gather*}
(\lambda.M)~N \bred \qba \derives M[N]
\qquad
\tlet(\bang_q M,N) \bred \qbb \derives N[q \derives M] 
\\
\bang_q \cF[\tinsp(\vartheta)] \bred \bang_q \cF[\qti \derives q\vartheta]
\end{gather*}
(trail inspection, which does not use substitutions, is unchanged). The idea is that explicit substitutions make reduction more efficient because their evaluation does not need to be performed all at once, but can be delayed, partially or completely; delayed explicit substitutions applied to the same term can be merged, so that the term does not need to be scanned twice. The evaluation of explicit substitution can be defined by the following $\sigma$-rules:
\[
\begin{array}{rl@{\qquad}rl}
1[\id] & \sred 1 & \id \comp s & \sred s\\
1[M \cdot s] & \sred M & \lift \comp \id & \sred {\lift} \\
(\lambda M)[s] & \sred \lambda (M[1\cdot (s \comp \lift)]) & \lift \comp (M \cdot s) & \sred s 
\\
(M~N)[s] & \sred M[s]~N[s] 
&
(M \cdot s)\comp t & \sred M[t] \cdot (s \comp t) 
\\
(\bang_q M)[s] & \sred \bang_q (M[s]) 
&
(s_1 \comp s_2) \comp s_3 & \sred s_1 \comp (s_2 \comp s_3) 
\\
\tlet(M,N)[s] & \sred \tlet(M,N[1\cdot (s \comp \lift)]) 
&
(q \triangleright M)[s] & \sred q \triangleright (M[s]) 
\\
\tinsp(\{ \vect{M} \})[s] & \sred \tinsp(\{\vect{M[s]}\})
&
M[s][t] & \sred M[s \comp t]
\end{array}
\]
These rules are a relatively minor adaptation from those of $\lambda\sigma$: as in that language, $\sigma$-normal forms do not contain explicit substitutions, save for the case of the index $1$, which may be lifted multiple times, e.g.:
\[
1[\lift^n] = 1[\underbrace{\lift \circ \cdots \circ \lift}_{\mbox{$n$ times}}]
\]
If we take $1[\lift^n]$ to represent the de~Bruijn index $n+1$, as in $\lambda\sigma$, $\sigma$-normal terms coincide with a nameless representation of \CAUm.

The $\sigma$-rules are deferrable, in that we can perform $\beta$-reductions even if a term is not in $\sigma$-normal form. We would like to treat the $\tau$-rules in the same way, perhaps performing $\tau$-normalization only before trail inspection; however, we can see that changing the order of $\tau$-rules destroys confluence even when $\beta$-redexes are triggered in the same order. 
\begin{figure}[t]
\small{
\[
\xymatrix{
(\lambda.M~1~1)~(q \derives N) \ar@[->>][r]^\tau \ar[d]_\beta & \qapp(\qrefl,q) \derives (\lambda.M~1~1)~N \ar[d]_\beta \\
\qba \derives (M~1~1)[q \derives N] \ar@{->>}[d]_{\sigma\tau} & 
\qapp(\qrefl,q) \derives \qba \derives (M~1~1)[N]
\ar@{->>}[d]_{\sigma\tau} \\
\qtrans(\qba,\qapp(\qapp(\qrefl,q),q)) \derives M~N~N  & \qtrans(\qapp(\qrefl,q),\qba) \derives M~N~N 
}
\]
}\caption{Non-joinable reduction in \CAUm with na\"ive explicit substitutions}\label{fig:nonjoinable}
\end{figure}
Consider for example the reductions in Figure~\ref{fig:nonjoinable}: performing a $\tau$-step before the beta-reduction, as in the right branch, yields the expected result. If instead we delay the $\tau$-step, the trail $q$ decorating $N$ is duplicated by beta reduction; furthermore, the order of $q$ and $\qba$ gets mixed up: even though $q$ records computation that happened (once) \emph{before} $\qba$, the final trail asserts that $q$ happened (twice) \emph{after} $\qba$.\footnote{Although the right branch describes an unfaithful account of history, it is still a coherent one: we will explain this in more detail in the conclusions.} 
As expected, the two trails (and consequently the terms they decorate) are not joinable.

The example shows that $\beta$-reduction on terms whose trails have not been normalized is \emph{anachronistic}. If we separated the trails stored in a term from the underlying, trail-less term, we might be able to define a \emph{catachronistic}, or time-honoring version of $\beta$-reduction. For instance, if we write $\erase{M}$ for trail-erasure and $\trailify{M}$ for the trail-extraction of a term $M$, catachronistic beta reduction could be written as follows:
\begin{align*}
(\lambda.M)~N & \bred \qtrans(\trailify{(\lambda.M)~N},\qba) \derives \erase{M}[\erase{N}] 
\\
\tlet(\bang_q M,N) & \bred \qtrans(\trailify{\tlet(\bang_q M,N)},\qbb) \derives \erase{N}[q \derives M]
\\
\bang_q \cF[\tinsp(\vartheta)] & \bred \bang_q \cF[\qti \derives q'\vartheta] \qquad \mbox{(where $q' = \tau(\qtrans(q,\trailify{\cF[\tinsp(\vartheta)]}))$)}
\end{align*}
\begin{techreport}
Without any pretense of being formal, we can give a partial definition
of trail-erasure $\erase{\cdot}$ and trail-extraction $\trailify{\cdot}$, which we collectively refer to as \emph{trail projections}, as follows:
\begin{align*}
\erase{1} & = 1 & \trailify{1} & = \qrefl \\
\erase{1[\lift^n]} & = 1[\lift^n] & \trailify{1[\lift^n]} & = \qrefl \\
\erase{\lambda.M} & = \lambda.\erase{M} & \trailify{\lambda.M} & = \qlam(\trailify{M}) \\
\erase{M~N} & = \erase{M}~\erase{N} & \trailify{M~N} & = \qapp(\trailify{M},\trailify{N}) \\
\erase{\bang_q M} & = \bang_q M & \trailify{\bang_q M} & = \qrefl \\
\erase{\tlet(M,N)} & = \tlet(\erase{M},\erase{N}) & \trailify{\tlet(M,N)} & = \qlet(\trailify{M},\trailify{N}) \\
\erase{q \derives M} & = \erase{M} & \trailify{q \derives M} & = \qtrans(q,\trailify{M}) \\
\erase{\tinsp(\{\vect{M}\})} & = \tinsp(\{ \vect{\erase{M}} \}) & 
\trailify{\tinsp(\{\vect{M}\})} & = \qtr(\{ \vect{\trailify{M}} \}) 
\end{align*}
This definition is only partial: we do not say what to do when the term contains explicit substitutions. When computing, say, $\trailify{M[s]}$, the best course of action we can think of is to obtain the $\sigma$-normal form of $M[s]$, which is a pure \CAUm term with no explicit substitutions, and then proceed with its trail-extraction.

But the whole approach is clumsy: trail-erasure and trail-extraction are multi-step operations that need to scan their entire argument, even when it does not contain any explicit substitution. We would achieve greater efficiency if they could be broken up into sub-steps, much like we did with substitution.

Surely, to obtain this result we need a language in which terms and trails can mention trail-erasure and trail-extraction explicitly. This is the language that we will introduce in the next section.
\end{techreport}
\begin{proceedings}
We could easily define trail erasure and extraction as operations on pure \CAUm terms (without explicit substitutions), but the cost of eagerly computing their result would be proportional to the size of the input term; furthermore, the extension to explicit substitutions would not be straightforward. Instead, in the next section, we will describe an extended language to manipulate trail projections explicitly.
\end{proceedings}
\section{The calculus \CAUms}\label{CAUms}
We define the untyped Calculus of Audited Units with explicit substitutions, or \CAUms, as the following extension of the syntax of \CAUm presented in Section~\ref{CAUm}:
\begin{tabular}{rcl}
$M,N$ & $::=$ & $1 \orelse \lambda.M \orelse M~N \orelse \tlet(M,N) \orelse \bang_q M \orelse q \triangleright M \orelse \tinsp(\vartheta) \orelse M[s] \orelse \erase{M}$ 
\\
$q,q'$ & $::=$ & $\qrefl \orelse \qtrans(q,q') \orelse \qba \orelse \qbb \orelse \qti \orelse \qlam(q) \orelse \qapp(q,q') \orelse \qlet(q,q') \orelse \qtr(\zeta) \orelse \trailify{M}$ \\
$s,t$ & $::=$ & $\id \orelse \lift \orelse M \cdot s \orelse s \comp t$
\end{tabular}

\medskip

\CAUms builds on the observations about explicit substitutions we made in the previous section: in addition to closures $M[s]$, it provides syntactic trail erasures denoted by $\erase{M}$; dually, the syntax of trails is extended with the explicit trail-extraction of a term, written $\trailify{M}$.
In the na\"ive presentation, we gave a satisfactory set of $\sigma$-rules defining the semantics of explicit substitutions, which we keep as part of \CAUms. 
To express the semantics of explicit projections, we provide in Figure~\ref{fig:sigmaexp} rules stating that $\erase{\cdot}$ and $\trailify{\cdot}$ commute with most term constructors (but not with $\bang$) and are blocked by explicit substitutions.
These rules are completed by congruence rules asserting that they can be used in any subterm or subtrail of a given term or trail.

\begin{figure}
\[
\begin{array}{rl@{\qquad}rl}
\erase{1} & \sred 1 & \trailify{1} & \sred \qrefl \\
\erase{1[\lift^n]} & \sred 1[\lift^n] & \trailify{1[\lift^n]} & \sred \qrefl\\
\erase{\lambda.M} & \sred \lambda.\erase{M} & \trailify{\lambda.M} & \sred \qlam(\trailify{M}) \\
\erase{M~N} & \sred \erase{M}~\erase{N} & \trailify{M~N} & \sred \qapp(\trailify{M},\trailify{N}) \\
\erase{\bang_q M} & \sred \bang_q M & \trailify{\bang_q M} & \sred \qrefl \\
\erase{\tlet(M,N)} & \sred \tlet(\erase{M},\erase{N}) & \trailify{\tlet(M,N)} & \sred \qlet(\trailify{M},\trailify{N}) \\
\erase{q \derives M} & \sred \erase{M} & \trailify{q \derives M} & \sred \qtrans(q,\trailify{M}) \\
\erase{\tinsp(\{\vect{M}\})} & \sred \tinsp(\{ \vect{\erase{M}} \}) &
\trailify{\tinsp(\{\vect{M}\})} & \sred \qtr(\{ \vect{\trailify{M}} \}) 
\end{array}
\]
\caption{$\sigma$-reduction for explicit trail projections}
\label{fig:sigmaexp}
\end{figure}

The $\tau$ rules from Section~\ref{CAUm} are added to \CAUms with the obvious adaptations. We prove that $\sigma$ and $\tau$, together, yield a terminating and confluent rewriting system.

\begin{theorem}\label{thm:ST}
$(\sred \cup \tred)$ is terminating and confluent.
\end{theorem}
\begin{proof}
Tools like AProVE~\cite{aprove} are able to prove termination
automatically. Local confluence can be proved easily
by considering all possible pairs of rules: full confluence follows as
a corollary of these two results.
\end{proof}

\subsection{Beta reduction}
We replace the definition of $\beta$-reduction by the following lazy
rules that use trail-extraction and trail-erasure to ensure that the
correct trails are eventually produced:
\begin{align*}
(\lambda.M)~N & 
\Bred \qtrans(\qapp(\qlam(\trailify{M}),\trailify{N}),\qba) \derives \erase{M}[\erase{N}] 
\\
\tlet(\bang_q M,N) & 
\Bred \qtrans(\qlet(\qrefl,\trailify{N}),\qbb) \derives \erase{N}[q \derives M] 
\\
\bang_q \cF[\tinsp(\vartheta)] & 
\Bred \bang_q \cF[\qti \derives q'\vartheta] \qquad \mbox{(where $q' = \sigma\tau(\qtrans(q,\trailify{\cF[\tinsp(\vartheta)]}))$)}
\end{align*}
where $\cF$ specifies that the reduction cannot take place within a bang, a substitution, or a trail erasure:
\begin{align*}
\cF & ::= \blacksquare \orelse \lambda.\cF \orelse (\cF~N) \orelse (M~\cF) \orelse \tlet(\cF,N) \orelse \tlet(M,\cF) \orelse q \derives \cF \orelse \tinsp(\vect{M},\cF,\vect{N}) \orelse \cF[s]
\end{align*}
As usual, the relation is extended to inner subterms by means of congruence rules. 
However, we need to be 
careful: we cannot reduce within a trail-erasure, because if we did, the newly created trail would be erroneously erased:
\[\begin{array}{rl}
\mbox{wrong:} & \erase{(\lambda.M)~N} 
\\
& \Bred \erase{\qtrans(\qapp(\qlam(\trailify{M}),\trailify{N}),\qba) \derives \erase{M}[\erase{N}]} 
\\
& \sred \erase{\erase{M}[\erase{N}]} \smallskip \\
\mbox{correct:} & \erase{(\lambda.M)~N} 
\\
& \srred (\lambda.\erase{M})~\erase{N} 
\\
& \Bred \qtrans(\qapp(\qlam(\trailify{\erase{M}}),\trailify{\erase{N}}),\qba) \derives \erase{M}[\erase{N}]
\end{array}\]
This is why we express the congruence rule by means of contexts $\cE_\sigma$ such that holes cannot appear within erasures (the definition also employs substitution contexts $\cS_\sigma$ to allow reduction within substitutions):
\begin{prooftree}
\AxiomC{$M \Bred N$}
\UnaryInfC{$\cE_\sigma[M] \Bred \cE_\sigma[N]$}
\end{prooftree}
Formally, evaluation contexts are defined as follows:
\begin{definition}[evaluation context]
\begin{align*}
\cE_\sigma & ::= \blacksquare \orelse \lambda.\cE_\sigma \orelse (\cE_\sigma~N) \orelse (M~\cE_\sigma) \orelse \tlet(\cE_\sigma,N) 
  \orelse \tlet(M,\cE_\sigma) \orelse \bang_q \cE_\sigma \orelse q \derives \cE_\sigma
\\
& \orelse \tinsp(\{ \vect{M},\cE_\sigma,\vect{N} \}) 
  \orelse \cE_\sigma[s] \orelse M[\cS_\sigma] 
\\
\cS_\sigma & ::= \cS_\sigma \circ t \orelse s \circ \cS_\sigma \orelse \cE_\sigma \cdot s \orelse M \cdot \cS_\sigma
\end{align*}
\label{def:evalctx}
\end{definition}
We denote $\sigma\tau$-equivalence (the reflexive, symmetric, and transitive closure of $\stred$) by means of $\steeq$. As we will prove, $\sigma\tau$-equivalent \CAUms terms can be interpreted as the same \CAUm term: for this reason, we define reduction in \CAUms as the union of $\Bred$ and $\steeq$:
\[
{\CAUmsred} \mathop{:=} {\Bred} \mathop{\cup} {\steeq}
\]

\subsection{Properties of the rewriting system}\label{properties}
The main results we prove concern the relationship between \CAUm and \CAUms: firstly, every \CAUm reduction must still be a legal reduction within \CAUms; in addition, it should be possible to interpret every \CAUms reduction as a \CAUm reduction over suitable $\sigma\tau$-normal terms.

\begin{theorem}\label{thm:simuR}
If $M \CAUmrred N$, then $M \CAUmsrred N$.
\end{theorem}
\begin{theorem}\label{thm:simu}
If $M \CAUmsrred N$, then $\sigma\tau(M) \CAUmrred \sigma\tau(N)$.
\end{theorem}

Although \CAUms, just like \CAUm, is \emph{not} confluent (different reduction strategies produce different trails, and trail inspection can be used to compute on them, yielding different terms as well), the previous results allow us to use Hardin's interpretation technique \cite{hardin89} to prove a \emph{relativized} confluence theorem:
\begin{theorem}
If $M \CAUmsrred N$ and $M \CAUmsrred R$, and furthermore $\sigma\tau(N)$ and $\sigma\tau(R)$ are joinable in \CAUm, then $N$ and $R$ are joinable in $\CAUms$.
\end{theorem}
\begin{proof}
See Figure~\ref{fig:confluence}.
\end{proof}

\begin{figure}[t]
\begin{center}
\(
\xymatrix{
& M \ar@{->>}[dl]_{\CAUms} \ar@{->>}[dr]^{\CAUms} \ar@{->>}[d]^{\sigma\tau} & \\
N \ar@{->>}[d]^{\sigma\tau} & \sigma\tau(M) \ar@{->>}[dl]^{\CAUm} \ar@{->>}[dr]_{\CAUm} & R \ar@{->>}[d]^{\sigma\tau}\\
\sigma\tau(N) \ar@{->>}[dr]_{\CAUm} & & \sigma\tau(R) \ar@{->>}[dl]^{\CAUm} \\
& S &
}
\)
\end{center}
\caption{Relativized confluence for $\CAUms$.}\label{fig:confluence}
\end{figure}

While the proof of Theorem~\ref{thm:simuR} is not overly different from the similar proof for the $\lambda\sigma$-calculus, Theorem~\ref{thm:simu} is more interesting. The main challenge is to prove that whenever $M \Bred N$, we have $\sigma\tau(M) \CAUmrred \sigma\tau(N)$. However, when proceeding by induction on $M \Bred N$, the terms $\sigma\tau(M)$ and $\sigma\tau(N)$ are too normalized to provide us with a good enough induction hypothesis: in particular, we would want them to be in the form $q \derives R$ even when $q$ is reflexivity. We call terms in this quasi-normal form \emph{focused}, and prove the theorem by reasoning on them. 
\begin{techreport}
The appendix contains the details of the proof.
\end{techreport}
\begin{proceedings}
The details of the proof are discussed in the extended version.
\end{proceedings}

\section{A call-by-value abstract machine}\label{SECD}
In this section, we 
derive an abstract machine implementing a weak call-by-value strategy.
More precisely, the machine will consider subterms shaped like $q \derives \erase{\pure{M}[e]}$, where $\pure{M}$ is a pure \CAUm term with no explicit operators, and $e$ is an \emph{environment}, i.e. an explicit substitution containing only values. 
In the tradition of lazy abstract machines, values are \emph{closures} (typically pairing a lambda and an environment binding its free variables); in our case, the most natural notion of closure also involves trail erasures and bangs:

\begin{tabular}{llcl}
\textbf{Closures} & $C$ & $::=$ & $\erase{(\lambda \pure{M})[e]} \orelse \bang_q C$ 
\\
\textbf{Values} & $V,W$ & $::=$ & $q \derives C$
\\
\textbf{Environments} & $e$ & $::=$ & $\id \orelse V \cdot e$
\end{tabular}

\medskip

According to this definition, the most general case of closure is a telescope of bangs, each equipped with a complete history, terminated at the innermost level by a lambda abstraction applied to an environment and enclosed in an erasure.
\[
\bang_{q_1} \cdots \bang_{q_n} \erase{(\lambda \pure{M})[e]}
\]
The environment $e$ contains values with dangling trails, which may be captured by bangs contained in $\pure{M}$; however, the erasure makes sure that none of these trails may reach the external bangs; thus, along with giving meaning to free variables contained in lambdas, closures serve the additional purpose of making sure the history described by the $q_1,\ldots,q_n$ is complete for each bang.

The machine we describe is a variant of the SECD machine. To simplify the description, the code and environment are not separate elements of the machine state, but they are combined, together with a trail, as the top item of the stack. Another major difference is that a code $\kappa$ can be not only a normal term without explicit operations, but also be a fragment of abstract syntax tree. The stack $\pi$ is a list of tuples containing a trail, a code, and an environment, and represents the subterm currently being evaluated (the top of the stack) and the unevaluated context, i.e. subterms whose evaluation has been deferred (the remainder of the stack). As a pleasant side-effect of allowing fragments of the AST into the stack, we never need to set aside the current stack into the dump: $D$ is just a list of values representing the evaluated context (i.e. the subterms whose evaluation has already been completed).

\begin{tabular}{llcl}
\textbf{Codes} & $\kappa$ & $::=$ & $\pure{M} \orelse @ \orelse \bang \orelse \tlet(\pure{M}) \orelse \tinsp$
\\
\textbf{Tuples} & $\tau$ & $::=$ & $(q|\kappa|e)$
\\
\textbf{Stack} & $\pi$ & $::=$ & $\vect{\tau}$
\\
\textbf{Dumps} & $D$ & $::=$ & $\vect{V}$
\\
\textbf{Configurations} & $\varsigma$ & $::=$ & $(\pi,D)$
\end{tabular}

\medskip

The AST fragments allowed in codes include application nodes $@$, bang nodes $\bang$, incomplete let bindings $\tlet(\pure{M})$, and inspection nodes $\tinsp$. A tuple $(q|\pure{M}|e)$ in which the code happens to be a term can be easily interpreted as $q \derives \erase{\pure{M}[e]}$; however, tuples whose code is an AST fragment only make sense within a certain machine state. The machine state is described by a configuration $\varsigma$ consisting of a stack and a dump.

\begin{figure}
\begin{center}
\begin{tabular}{cc}

\begin{minipage}{0.43\textwidth}

\begin{prooftree}
\AxiomC{$\jctx{(\epsilon,\epsilon)}$}
\end{prooftree}

\begin{prooftree}
\AxiomC{$\jctx{(\pi,D)}$}
\UnaryInfC{$\jctx{((q|\pure{M}|e)::(q'|@|\id)::\pi,D)}$}
\end{prooftree}

\begin{prooftree}
\AxiomC{$\jctx{(\pi,D)}$}
\UnaryInfC{$\jctx{((q|@|\id)::\pi,V::D)}$}
\end{prooftree}

\begin{prooftree}
\AxiomC{$\jctx{(\pi,D)}$}
\UnaryInfC{$\jctx{((q|\tlet(\pure{M})|e)::\pi,D)}$}
\end{prooftree}

\begin{prooftree}
\AxiomC{$\jctx{(\pi,D)}$}
\UnaryInfC{$\jctx{((q|\bang|\id)::\pi,D)}$}
\end{prooftree}

\begin{prooftree}
\AxiomC{$\jctx{(\pi,D)}$}
\UnaryInfC{$\jctx{\left(
	\begin{array}{l}
	\vect{(q_i|\pure{M_i}|e_i)_{i = k+1,\ldots,9}}::
	\\
	\quad(q'|\tinsp|\id)::\pi, \\
	\vect{V_{\{j = 1,\ldots,k-1\}}}::D
	\end{array}
	\right)
}$}
\end{prooftree}

\end{minipage}

&

\begin{minipage}{0.43\textwidth}

\begin{prooftree}
\AxiomC{$\jtm{(\epsilon,V::\epsilon)}$}
\end{prooftree}

\begin{prooftree}
\AxiomC{$\jctx{(\pi,D)}$}
\UnaryInfC{$\jtm{((q|\pure{M}|e)::\pi,D)}$}
\end{prooftree}

\begin{prooftree}
\AxiomC{$\jctx{(\pi,D)}$}
\UnaryInfC{$\jtm{((q|@|\id)::\pi,W::V::D)}$}
\end{prooftree}

\begin{prooftree}
\AxiomC{$\jctx{(\pi,D)}$}
\UnaryInfC{$\jtm{((q|\tlet(\pure{M})|e)::\pi,V::D)}$}
\end{prooftree}

\begin{prooftree}
\AxiomC{$\jctx{(\pi,D)}$}
\UnaryInfC{$\jtm{((q|\bang|\id)::\pi,V::D)}$}
\end{prooftree}

\begin{prooftree}
\AxiomC{$\jctx{(\pi,D)}$}
\UnaryInfC{$\jctx{((q|\tinsp|\id)::\pi,\vect{V_9}::D)}$}
\end{prooftree}
\end{minipage}

\end{tabular}
\end{center}
\caption{Term and context configurations}
\label{fig:wfcfg}
\end{figure}

A meaningful state cannot contain just any stack and dump, but must have a certain internal coherence, which we express by means of the two judgments in Figure~\ref{fig:wfcfg}: in particular, the machine state must be a \emph{term configuration}; this notion is defined by the judgment $\jtm{\varsigma}$, which employs a separate notion of \emph{context configuration}, described by the judgment $\jctx{\varsigma}$.

We can define the denotation of configurations by recursion on their well-formedness judgment:

\begin{definition}\mbox{}

\begin{enumerate}
\item The denotation of a context configuration is defined as follows:
\begin{align*}
\dec{(\epsilon,\epsilon)} & \triangleq \Hole
\\
\dec{((q|\pure{M}|e)::(q'|@|\id)::\pi,D)} & \triangleq \dec{(\pi,D)}[q' \derives (\Hole~(q \derives \erase{\pure{M}[e]}))]
\\
\dec{((q|@|\id)::\pi,V::D)} & \triangleq \dec{(\pi,D)}[q \derives (V~\Hole)]
\\
\dec{((q|\tlet(\pure{M})|e)::\pi,D)} & \triangleq \dec{(\pi,D)}[q \derives \tlet(\Hole, \erase{\pure{M}[1 \cdot (e \circ \uparrow)]})]
\\
\dec{((q|\bang|\id)::\pi,D)} & \triangleq \dec{(\pi,D)}[q \derives \bang \Hole]
\\
\dec{(\vect{(q_i|\pure{M_i}|e_i)}::(q'|\tinsp|\id)::\pi,\vect{V_j}::D)} & \triangleq \dec{(\pi,D)}[q' \derives \tinsp(\vect{V_j},\Hole,\vect{(q_i \derives \erase{\pure{M_i}[e_i]})})]
\end{align*}
where in the last line $i + j + 1 = 9$.

\item The denotation of a term configuration is defined as follows:
\begin{align*}
\tmdec{\epsilon,V::\epsilon} & \triangleq V
\\
\tmdec{(q|\pure{M}|e)::\pi,D} & \triangleq \dec{(\pi,D)}[q \derives \erase{\pure{M}[e]}]
\\
\tmdec{(q|@|\id)::\pi,W::V::D} & \triangleq \dec{(\pi,D)}[q \derives (V~W)]
\\
\tmdec{(q|\tlet(\pure{M})|e)::\pi,V::D} & \triangleq \dec{(\pi,D)}[q \derives \tlet(V, \erase{\pure{M}[1 \cdot (e \circ \uparrow)])]}
\\
\tmdec{(q|\bang|\id)::\pi,V::D} & \triangleq \dec{(\pi,D)}[q \derives \bang V]
\\
\tmdec{(q|\tinsp|\id)::\pi,\vect{V_9}::D} & \triangleq \dec{(\pi,D)}[q \derives \tinsp(\vect{V_9})]
\end{align*}
\end{enumerate}
\end{definition}

We see immediately that the denotation of a term configuration is a \CAUms term, while that of a context configuration is a \CAUms context (Definition~\ref{def:evalctx}).

\begin{figure}[t]
\hspace{-1.2cm}
\scriptsize{
\( \displaystyle
\begin{array}{|r||c|c||c|c|}
\hline
&
\multicolumn{2}{c@{\mapsto}}{
\mbox{source}
}
&
\multicolumn{2}{c|}{
\mbox{target}
}
\\
\hline
1 &
(q|\pure{M}~\pure{N}|e)::\pi & D 
& 
(\qrefl|\pure{M}|e)::(\qrefl|\pure{N}|e)::(q|@|\id)::\pi & D
\\
2 &
(q|@|\id)::\pi & (q' \derives C)::(q'' \derives \erase{(\lambda \pure{M})[e]})::D
& 
(q;\qapp(q'',q');\qba|\pure{M}|(\qrefl \derives C) \cdot e)::\pi & D
\\
3 &
(q|\lambda \pure{M}|e)::\pi & D
&
\pi & (q \derives \erase{(\lambda \pure{M})[e]})::D
\\
4 &
(q|\tlet(\pure{M},\pure{N})|e)::\pi & D
&
(\qrefl|\pure{M}|e)::(q|\tlet(\pure{N})|e)::\pi & D
\\
5 &
(q|\tlet(\pure{N})|e)::\pi & (q' \derives \bang V)::D
&
(q;\qlet(q',\qrefl);\qbb;q_{\pure{N},e,V}|\pure{N}|V \cdot e)::\pi & D
\\
6 &
(q|\bang_{q'} \pure{M}|e)::\pi & D
&
(q';\trailify{\pure{M}[e]}|\pure{M}|e)::(q|\bang|\id)::\pi & D
\\
7 &
(q|\bang|\id)::\pi & V::D
&
\pi & (q \derives \bang V)::D
\\
8 &
(q|\tinsp(\vect{\pure{M_9}})|e)::\pi & D
&
\vect{(\qrefl|\pure{M_i}|e)_{i=1,\ldots,9}}::(q|\tinsp|\id)::\pi & D
\\
9 &
(q|\tinsp|\id)::\pi & \vect{(q_i \derives C_i)_{i=1,\ldots,9}}::D
&
(q;\qtr(\vect{q_i});\qti|J_{q,\vect{q_i},\pi,D}|[\vect{(\qrefl \derives C_i)}])::\pi & D
\\
10 &
(q|n|e)::\pi & D 
&
\pi & (q \derives e(n))::D
\\
\hline
\end{array}
\)
}
\[
\begin{array}{rl}
q_{\pure{N},e,V} & \triangleq \trailify{\erase{\pure{N}[1\cdot(e \circ \uparrow)]}[V]}
\\
J_{q,\vect{q_i},\pi,D} & \triangleq \cI((q;\qtr(\vect{q_i})),\pi,D)
\\
e(n) & \triangleq \left\{
\begin{array}{ll}
C & \mbox{if $e = (q \derives C) \cdot e'$ and $n = 1$}
\\
e'(m) & \mbox{if $e = V \cdot e'$ and $n = m + 1$}
\end{array}
\right.
\end{array}
\]
\caption{Call-by-value abstract machine}
\label{fig:machine}
\end{figure}

\begin{figure}[b]
\begin{align*}
\cI(q_\vartheta,(q'|\bang|\epsilon)::\pi,D) = & 
\sigma\tau(q_\vartheta)
\\
\cI(q_\vartheta,(q'|\pure{M}|e)::(q''|@|\epsilon)::\pi,D) = & \cI((q'';\qapp(q_\vartheta,q')),\pi,D)
\\
\cI(q_\vartheta,(q'|@|\id)::\pi,(q'' \derives C)::D) = & \cI((q',\qapp(q'',q_\vartheta)),\pi,D)
\\
\cI(q_\vartheta,(q'|\tlet(\pure{M})|e)::\pi,D) = & \cI((q';\qlet(q_\vartheta,\qrefl)),\pi,D)
\\
\cI(q_\vartheta,\vect{(q_i|M_i|e_i)}::(q'|\tinsp|\id)::\pi,\vect{(q_j \derives C_j)}::D) = & 
\cI((q';\qtr(\vect{q_j},q_\vartheta,\vect{q_i})),\pi,D)
\end{align*}
\caption{Materialization of trails for inspection}
\label{fig:materialize}
\end{figure}

The call-by-value abstract machine for \CAUm is shown in Figure~\ref{fig:machine}: in this definition we use semi-colons as a compact notation for sequences of transitivity trails. The evaluation of a pure, closed term $\pure{M}$, starts with an empty dump and a stack made of a single tuple $(\qrefl,\pure{M},\id)$: this is a term configuration denoting $\qrefl \derives \erase{\pure{M}[\id]}$, which is $\sigma\tau$-equivalent to $\pure{M}$. Final states are in the form $\epsilon,V::\epsilon$, which simply denotes the value $V$. When evaluating certain erroneous terms (e.g. $(\bang~M)~V$, where function application is used on a term that is not a function), the machine may get stuck in a non-final state; these terms are rejected by the typed \CAU.
The advantage of our machine, compared to a naive evaluation strategy, is that in our case all the principal reductions can be performed in constant time, except for trail inspection which must examine a full trail, and thus will always require a time proportional to the size of the trail.

Let us now examine the transition rules briefly. Rules 1-3 and 10 closely match the ``Split CEK'' machine~\cite{Accattoli2014} (a simplified presentation of the SECD machine), save for the use of the $@$ code to represent application nodes, while in the Split CEK machine they are expressed implicitly by the stack structure. 

Rule 1 evaluates an application by decomposing it, placing two new tuples on the stack for the subterms, along with a third tuple for the application node; the topmost trail remains at the application node level, and two reflexivity trails are created for the subterms; the environment is propagated to the subterm tuples. 

The idea is that when the machine reaches a state in which the term at the top of the stack is a value (e.g. a lambda abstraction, as in rule 3), the value is moved to the dump, and evaluation continues on the rest of the stack. Thus when in rule 2 we evaluate an application node, the dump will contain two items resulting from the evaluation of the two subterms of the application; for the application to be meaningful, the left-hand subterm must have evaluated to a term of the form $\lambda \pure{M}$, whereas the form of the right-hand subterm is not important: the evaluation will then continue as usual on $\pure{M}$ under an extended environment; the new trail will be obtained by combining the three trails from the application node and its subexpressions, followed by a $\qba$ trail representing beta reduction.

The evaluation of $\tlet$ works similarly to that of applications; however, a term $\tlet(\pure{M},\pure{N})$ is split intro $\pure{M}$ and $\tlet(\pure{N})$ (rule 4), so that $\pure{N}$ is never evaluated independently from the corresponding $\tlet$ node. When in rule 5 we evaluate the $\tlet(\pure{N})$ node, the dump will contain a value corresponding to the evaluation of $\pure{M}$ (which must have resulted in a value of the form $\bang V$): we then proceed to evaluate $\pure{N}$ in an environment extended with $V$; this step corresponds to a principal contraction, so we update the trail accordingly, by adding $\qbb$; additionally, we need to take into account the trails from $V$ after substitution into $\pure{N}$: we do this by extending the trail with $\trailify{\erase{\pure{N}[1 \cdot (e \circ \uparrow)]}[V]}$.

Bangs are managed by rules 6 and 7. To evaluate $\bang_{q'} \pure{M}$, we split it into $\pure{M}$ and a $\bang$ node, placing the corresponding tuples on top of the stack; the original external trail $q$ remains with the $\bang$ node, whereas the internal trail $q'$ is placed in the tuple with $\pure{M}$; the environment $e$ is propagated to the body of the bang but, since it may contain trails, we need to extend $q'$ with the trails resulting from substitution into $\pure{M}$. When in rule 7 we evaluate the $\bang$ node, the top of the dump contains the value $V$ resulting from the evaluation of its body: we update the dump by combining $V$ with the bang and proceed to evaluate the rest of the stack.

The evaluation of trail inspections (rules 8 and 9) follows the same principle as that of applications, with the obvious differences due to the fact that inspections have nine subterms. The principal contraction happens in rule 9, which assumes that the inspection branches have been evaluated to $q_1 \derives C_1,\ldots,q_9 \derives C_9$ and put on the dump: at this point we have to reconstruct and normalize the inspection trail and apply the inspection branches. To reconstruct the inspection trail, we combine $q$ and the $\vect{q_i}$ into the trail for the current subterm $(q;\qtr(\vect{q_i}))$; then we must collect the trails in the context of the current bang, which are scattered in the stack and dump: this is performed by the auxiliary operator $\cI$ of Figure~\ref{fig:materialize}, defined by recursion on the well-formedness of the context configuration $\pi,D$; the definition is partial, as it lacks the case for $\epsilon,\epsilon$, corresponding to an inspection appearing outside all bangs: such terms are considered ``inspection-locked'' and cannot be reduced.
Due to the operator $\cI$, rule 9 is the only rule that cannot be performed in constant time.

$\cI$ returns a $\sigma\tau$-normalized trail, which we need to apply to the branches $C_1,\ldots,C_9$; from the implementation point of view, this operation is analogous to a substitution replacing the trail nodes ($\qrefl, \qtrans, \qba, \qapp, \qlam, \ldots$) with the respective $M_i$. Suppose that trails are represented as nested applications of dangling de~Bruijn indices from $1$ to $9$ (e.g. the trail $\qapp(\qrefl,\qba)$ can be represented as $(1~2~3)$ for $\qapp=1$, $\qrefl=2$ and $\qba=3$); then trail inspection reduction amounts to the evaluation of a trail in an environment composed of the trail inspection branches. To sum it up, rule 9 produces a state in which the current tuple contains:
\begin{itemize}
\item a trail $(q;\qtr(\vect{q_i});\qti)$ (combining the trail of the inspection node, the trails of the branches, and the trail $\qti$ denoting trail inspection
\item the $\sigma\tau$-reduced inspection ``trail'' (operationally, an open term with nine dangling indices) which results from $\cI((q;\qtr(\vect{q_i})),\pi,D)$
\item an environment $[\vect{(\qrefl \derives C_i)}]$ which implements trail inspection by substituting the inspection branches for the dangling indices in the trail.
\end{itemize}

The machine is completed by rule 10, which evaluates de~Bruijn indices by looking them up in the environment. Notice that the lookup operation $e(n)$, defined when the de~Bruijn index $n$ is closed by the environment $e$, simply returns the $n$-th closure in $e$, but \emph{not} the associated trail; the invariants of our machine ensure that this trail is considered elsewhere (particularly in rules 5 and 6).

The following theorem states that the machine correctly implements reduction.
\begin{theorem}\label{thm:SECDsound}
For all valid $\varsigma$, $\varsigma \mapsto \varsigma'$ implies $\tmdec{\varsigma} \CAUmsrred \tmdec{\varsigma'}$.
\end{theorem}

\section{Conclusions and Future Directions}\label{sec:concl}
The calculus \CAUms which we introduced in this paper provides a finer-grained view over the reduction of history-carrying terms, 
and proved an effective tool in the study of the smarter evaluation techniques which we implemented in an abstract machine.
\CAUms is not limited to the call-by-value strategy used by our machine, and in future work we plan to further our investigation of efficient auditing to call-by-name and call-by-need.
%
%
Another intriguing direction we are exploring is to combine our
approach with recent advances in explicit substitutions, such as the
linear substitution calculus of Accattoli and
Kesner~\cite{Accattoli2010}, and apply the \emph{distillation}
technique of Accattoli et al.~\cite{Accattoli2014} 

In our discussion, we showed that the original definition of beta-reduction, when applied to terms that are not in trail-normal form, creates temporally unsound trails. We might wonder whether these anachronistic trails carry any meaning: let us take, as an example, the reduction on the left branch of Figure~\ref{fig:nonjoinable}:
\[
(\lambda.M~1~1)~(q \derives N) \rred{} \qtrans(\qba,\qapp(\qapp(\qrefl,q),q)) \derives M~N~N
\]
We know that $q$ is the trace left behind by the reduction that led to $N$ from the original term, say $R$:
\[
R \longrightarrow q \derives N
\]
We can see that the anachronistic trail is actually consistent with the reduction of $(\lambda.M~1~1)~R$ under a leftmost-outermost strategy:
\begin{align*}
& (\lambda.M~1~1)~R \longrightarrow \qba \derives M~R~R \rred{} \qba \derives M~(q\derives N)~(q\derives N) \\
& \qquad \rred{} \qtrans(\qba,\qapp(\qapp(\qrefl,q),q)) \derives M~N~N
\end{align*}
Under the anachronistic reduction, $q$ acts as the witness of an original inner redex. Through substitution within $M$, we get evidence that the contraction of an inner redex can be swapped with a subsequent head reduction: this is a key result in the proof of standardization that is usually obtained using the notion of \emph{residual} (\cite{Barendregt}, Lemma~11.4.5). Based on this remark, we conjecture that trails might be used to provide a more insightful proof: it would thus be interesting to see how trails relate to recent advancements in standardization~(\cite{AccattoliNonstandard,Asperti2013,Xi1999,Kashima2000}).

\paragraph{Acknowledgments.}
Effort sponsored by the Air Force Office of Scientific Research, Air Force Material Command, USAF, under grant number FA8655-13-1-3006. The U.S. Government and University of Edinburgh are authorised to reproduce and distribute reprints for their purposes notwithstanding any copyright notation thereon. Cheney was also supported by ERC Consolidator Grant Skye (grant number 682315). We are grateful to James McKinna and the anonymous reviewers for comments.

\bibliographystyle{splncs04}
\bibliography{expsub-cau}

\newpage
\appendix

\section{Proofs about \CAUms}
In this appendix, we provide more detail about the proofs mentioned in Section~\ref{properties}. A first useful result, expressed by the following lemma, shows that $\sigma$-normal terms coincide with the nameless
variant of $\CAUm$ outlined in Section~\ref{naive}:
\begin{lemma}
The $\sigma$-normal terms, trails and substitutions of \CAUms are expressed by the following grammar:
\[
\begin{array}{rl}
M, N & ::= 1 \orelse 1[\lift^n] \orelse \lambda.M \orelse M~N \orelse \tlet(M,N) \orelse \bang_q M \orelse q \derives M \orelse \tinsp(\vartheta) \\
q, q' & ::= \qrefl \orelse \qtrans(q,q') \orelse \qba \orelse \qbb \orelse \qti \orelse \qlam(q) \orelse \qapp(q, q') \orelse \qlet(q,q') \orelse \qtr(\zeta) 
\\
s, t & ::= \id \orelse \lift^n \orelse M \cdot s
\end{array}
\]
\end{lemma}

\noindent From this result, we extract a definition of $\sigma$-normal contexts $\cE$ and $\cS$:
\begin{align*}
\cE & ::= \blacksquare \orelse \lambda.\cE \orelse (\cE~N) \orelse (M~\cE) \orelse \tlet(\cE,N) \orelse \tlet(M,\cE) \orelse \bang_q \cE \orelse q \derives \cE 
\orelse \tinsp(\{ \vect{M},\cE,\vect{N} \}) 
\\
\cS & ::= \cE \cdot s \orelse M \cdot \cS 
\end{align*}
where all the terms $M,N$, trails $q$ and substitutions $s$ appearing in these definitions are $\sigma$-normal.

We define meta-level projections and focused forms, and then we prove some theorems about them.

\begin{definition}
The meta-level projections $\Erase{M}$ and $\Trailify{M}$ are defined as follows:
\[
\Erase{M} = 
  \begin{cases}
  M' & \mbox{if $\sigma\tau(M) = q \derives M'$} \\
  \sigma\tau(M) & \mbox{else}
  \end{cases}
\qquad
\Trailify{M} = 
  \begin{cases}
  q & \mbox{if $\sigma\tau(M) = q \derives M'$} \\
  \qrefl & \mbox{else}
  \end{cases}
\]
\end{definition}
\begin{definition}
The focused form of a term $M$, denoted by $\focus{M}$, is defined by:
\[
\focus{M} = \Trailify{M} \derives \Erase{M}
\]
The focused form of a $\sigma$-normal substitution $s = \vect{N}\cdot\lift^p$ is defined by:
\[
\focus{s} = \vect{\focus{N}}\cdot\lift^p
\]
This definition is extended to all substitutions by taking $\focus{s} = \focus{\sigma(s)}$.
\end{definition}

\begin{lemma}\label{lem:focusjoin}
For all $M$ and $s$, we have $\sigma\tau(M) = \sigma\tau(\focus{M})$ and $\sigma\tau(s) = \sigma\tau(\focus{s})$.
\end{lemma}
\begin{lemma}\label{lem:projF2R}
For all $M$, $\erase{M} \strred \Erase{M}$ and $\trailify{M} \strred \Trailify{M}$.
\end{lemma}
\begin{proof}
After unfolding the definitions, the proof is by induction on $\sigma\tau(M)$.
\end{proof}

We also define a meta-level operation corresponding to explicit substitutions, and prove the correspondence:
\begin{definition}
For $\sigma$-normal $M$ and $\vect{N} = N_1,\ldots,N_k$, the meta-level substitution $M\ssubst{p}{\vect{N}}$ is defined by recursion on $\sigma$-normal terms as follows:
\begin{align*}
m\ssubst{p}{\vect{N}} & = N_m & n\ssubst{p}{\vect{N}} & = n+p \\
(\lambda.M)\ssubst{p}{\vect{N}} & = \lambda.(M\ssubst{p+1}{1,\vect{\tlift(N)}}) & (R~S)\ssubst{p}{\vect{N}} & = R\ssubst{p}{\vect{N}} ~ S\ssubst{p}{\vect{N}} \\
\tlet(R,S)\ssubst{p}{\vect{N}} & = \tlet\left(R\ssubst{p}{\vect{N}},\right.  &
(\bang_q M)\ssubst{p}{\vect{N}} & = \bang_q (M\ssubst{p}{\vect{N}}) \\
& \qquad \left.S\ssubst{p+1}{1,\vect{\tlift(N)}}\right) & 
\tinsp(\{\vect{M} \})\ssubst{p}{\vect{N}} & = \tinsp(\{ \vect{M\ssubst{p}{\vect{N}}} \}) \\
(q \derives M)\ssubst{p}{\vect{N}} & = q \derives
                                     (M\ssubst{p}{\vect{N}})
\end{align*}
where $m \leq k < p$ and $\vect{\tlift(N)}$ lifts by one all of the free indices in each $N_i$. \\
We will write $M\ssubst{}{\vect{N}}$ as syntactic sugar for $M\ssubst{0}{\vect{N}}$. 
\end{definition}
\begin{lemma}\label{lem:subF2R}
For all $M,\vect{N}$ in $\sigma\tau$-normal form, $M[\vect{N} \cdot \lift^p] \strred M \ssubst{p}{\vect{N}}$. In particular, $M[R] \strred M\ssubst{}{R}$.
\end{lemma}
\begin{proof}
Routine induction on $\sigma$-normal forms $M$.
\end{proof}

We can use meta-substitution to define $\beta$-reduction for $\sigma$-normal terms in the same style as in \CAUm, and lift it to eager $\bar{\beta}$-reduction acting on focused terms:
\begin{center}
\begin{tabular}{c}
$(\lambda.M)~N \bred \qba \derives M\ssubst{}{N} \qquad
\tlet(\bang_q M,N) \bred \qbb \derives N\ssubst{}{q \derives M}$ \\
\rVS
$\bang_q \cF[\tinsp(\vartheta)] \bred \bang_q \cF[\qti \derives q\vartheta]
\qquad$ \\
\rVS
\AxiomC{$M \bred N$}
\UnaryInfC{$\cE[M] \bred \cE[N]$}
\DisplayProof
$\qquad$
\AxiomC{$M \bred N$}
\UnaryInfC{$\cS[M] \bred \cS[N]$}
\DisplayProof
\\
\rVS
\AxiomC{$M \bred N$}
\UnaryInfC{$\focus{M} \bbred \focus{N}$}
\DisplayProof
$\qquad$
\AxiomC{$s \bred t$}
\UnaryInfC{$\focus{s} \bbred \focus{t}$}
\DisplayProof
\end{tabular}
\end{center}
The last two rules, defining $\bar{\beta}$-reduction in terms of $\beta$-reduction, are not sufficiently compositional for our proofs, but we can prove a version where the premise is also a $\bar{\beta}$-reduction:
\begin{lemma}\label{lem:bbredCtx}
If $M \bbred N$, then $\focus{\cE[M]} \bbred \focus{\cE[N]}$ and $\focus{\cS[M]} \bbred \focus{\cS[N]}$.
\end{lemma}
\begin{proof}
From the hypothesis $M \bbred N$ we obtain $M',N'$ such that $M = \focus{M'}$, $N = \focus{N'}$ and $M' \bred N'$. We proceed by induction on this reduction: in the congruence case we obtain a context $\cE'$ that we need to merge with $\cE$ to obtain the thesis. The substitution case follows similarly.
\end{proof}

\noindent Then we prove that $\beta$-reduction can be simulated in \CAUms:
\begin{lemma}\label{lem:bred2BredStred}
If $M \bred N$, then $M \Bred \strred N$
\end{lemma}
\begin{proof}
By induction on the hypothesis, using Lemma~\ref{lem:subF2R}.
\end{proof}

\subsubsection*{Proof of Theorem~\ref{thm:simuR}.} 
\emph{If $M \CAUmrred N$, then $M \CAUmsrred N$.}

\noindent $M \CAUmred N$ implies $M \bred R \trred \tau(R) = N$ for some $R$; then $M \CAUmsrred N$ is an immediate consequence of Lemma~\ref{lem:bred2BredStred}.

\medskip

We are also interested in proving the dual statement: for every \CAUms
$\mbox{Beta}$-reduction there should be a \CAUm reduction on the corresponding $\sigma\tau$-normal forms. To prove this theorem, we need a suitable version of the usual substitutivity property of $\beta$-reduction, which is proved in the standard way:
\begin{lemma}\label{lem:CAUsubstlemma}
\mbox{}\\
If $M \bred M'$, then $M \ssubst{k}{\vect{N}} \bred M' \ssubst{k}{\vect{N}}$.\\
If $N_i \bred N'_i$, then $M \ssubst{k}{N_1\cdots N_{i-1},N_i,N_{i+1} \cdots N_p} \brred \\
M\ssubst{k}{N_1\cdots N_{i-1},N'_i,N_{i+1} \cdots N_p}$.
\end{lemma}

We can now prove that $\mbox{Beta}$-reductions can be mimicked by $\bar{\beta}$-reductions on the corresponding focused forms:
\begin{lemma}\label{lem:betaF2R}
If $M \Bred N$ then $\focus{M} \bbrred \focus{N}$.\\
If $s \Bred t$ then $\focus{s} \bbrred \focus{t}$.
\end{lemma}
\begin{proof}
By mutual induction on the derivations of $M \Bred N$ and $s \Bred t$. We consider two base cases, the applied explicit substitution case, and one of the other inductive cases (reduction in the first subterm of an application); the remaining cases can be proved similarly.

The first base case consists of lambda-application redexes:
\[
(\lambda M)~N \Bred \qtrans(\qapp(\qlam(\trailify{M}),\trailify{N}),\qba) \derives \erase{M}[\erase{N}]
\]
We need to prove that:
\[
\focus{(\lambda M)~N} \bbred \focus{\qtrans(\qapp(\qlam(\trailify{M}),\trailify{N}),\qba)) \derives \erase{M}[\erase{N}]}
\]
By confluence of $\sigma\tau$, we rewrite the left-hand side of the thesis with a $\sigma\tau$-equivalent focused term:
\[
\focus{(\lambda M)~N)} = \focus{\qapp(\qlam(\Trailify{M}),\Trailify{N}) \derives
  (\lambda \Erase{M})~\Erase{N}}
\]
Now we perform a $\bbred$ step to obtain:
\begin{align*}
& \focus{\qapp(\qlam(\Trailify{M}),\Trailify{N}) \derives \qba \derives \Erase{M}\ssubst{}{\Erase{N}}} \\
= & \focus{\qtrans(\qapp(\qlam(\Trailify{M}),\Trailify{N}),\qba) \derives \Erase{M}\ssubst{}{\Erase{N}}}
\end{align*}
We prove that the rhs of the thesis equals the result of the $\bar{\beta}$-reduction:
\begin{align*}
& \focus{\qtrans(\qapp(\qlam(\trailify{M}),\trailify{N}),\qba) \derives \erase{M}[\erase{N}]} \\
= & \focus{\qtrans(\qapp(\qlam(\Trailify{M}),\Trailify{N}),\qba) \derives \Erase{M}[\Erase{N}]} \\
= & \focus{\qtrans(\qapp(\qlam(\Trailify{M}),\Trailify{N}),\qba) \derives \Erase{M}\ssubst{}{\Erase{N}}}
\end{align*}

\noindent The case of let-box redexes is similar. In \CAUms, we have:
\[
\tlet(\bang_q M,N) \Bred \qtrans(\qlet(\qrefl,\trailify{N}),\qbb) \derives \erase{N}[q \derives M]
\]
We need to prove that:
\[
\focus{\tlet(\bang_q M,N)} \bred \focus{\qtrans(\qlet(\qrefl,\trailify{N}),\qbb) \derives \erase{N}[q \derives M]}
\]
We rewrite the lhs of the thesis:
\[
\focus{\tlet(\bang_q M,N)} = \focus{\qlet(\qrefl,\Trailify{N}) \derives \tlet(\bang_q M, \Erase{N})}
\]
then perform a $\bar{\beta}$-step to obtain:
\begin{align*}
& \focus{\qlet(\qrefl,\Trailify{N}) \derives \qbb \derives \Erase{N}\ssubst{}{q \derives M}} \\
= & \focus{\qtrans(\qlet(\qrefl,\Trailify{N}), \qbb) \derives \Erase{N}\ssubst{}{q \derives M}} \\
= & \focus{\qtrans(\qtrans(\qlet(\qrefl,\Trailify{N}),\qbb), \Trailify{\Erase{N}\ssubst{}{q \derives M}}) \derives \Erase{\Erase{N}\ssubst{}{q \derives M}}}
\end{align*}
We prove that the rhs of the thesis equals the result of the $\bar{\beta}$-reduction:
\begin{align*}
& \focus{\qtrans(\qlet(\qrefl,\trailify{N}),\qbb) \derives \erase{N}[q \derives M]} \\
= & \focus{\qtrans(\qlet(\qrefl,\Trailify{N}),\qbb) \derives \Erase{N}[q \derives M])} \\
= & \focus{\qtrans(\qlet(\qrefl,\Trailify{N}),\qbb) \derives \Erase{N}\ssubst{}{q \derives M}} \\
= & \focus{\qtrans(\qtrans(\qlet(\qrefl,\Trailify{N}),\qbb),\Trailify{\Erase{N}\ssubst{}{q \derives M}}) \derives \Erase{\Erase{N}\ssubst{}{q \derives M}}}
\end{align*}

In the application case, we need to prove $\focus{M~N} \bbred \focus{M'~N}$ under the induction hypothesis that $\focus{M} \bbred \focus{M'}$. By Lemma~\ref{lem:bbredCtx} we prove
\[
\focus{\focus{M}~N} \bbred \focus{\focus{M'}~N}
\]
which equals the thesis by substitution for $\sigma\tau$-equivalent subterms.

For applied explicit substitution, we have the following cases:
\begin{gather*}
\focus{M} \bbrred \focus{M'} \Longrightarrow \focus{M[s]} \bbrred \focus{M'[s]} \\
\focus{s} \bbrred \focus{s'} \Longrightarrow \focus{M[s]} \bbrred \focus{M[s']} 
\end{gather*}
They are both consequences of Lemma~\ref{lem:CAUsubstlemma}.
\end{proof}

In the following results, we write $\CAUmred$ for the reduction relation in $\CAUm$, i.e. a $\bred$-step followed by $\tau$-normalization, and $\CAUmsred$ for the full rewriting system of $\CAUms$.
\begin{lemma}\label{lem:bbred2bred}
If $M \bbred N$ then $\sigma\tau(M) \CAUmred \sigma\tau(N)$.
\end{lemma}
\begin{proof}
By the definition of $\bar{\beta}$-reduction, we obtain $M = \focus{M'}, N = \focus{N'}$ such that $M' \bred N'$, then we prove $\sigma\tau(\focus{M'}) = M'$ and $\sigma\tau(\focus{N'}) = \tau(N')$.
\end{proof}
\begin{lemma}\label{lem:simuBred}
If $M \Bred N$, then $\sigma\tau(M) \CAUmrred \sigma\tau(N)$
\end{lemma}
\begin{proof}
By Lemma~\ref{lem:betaF2R} we get $\focus{M} \bbrred \focus{N}$. By Lemma~\ref{lem:bbred2bred} and Lemma~\ref{lem:focusjoin}, we prove the thesis.
\end{proof}
\begin{lemma}\label{lem:simuStred}
If $M \steeq N$ then $\sigma\tau(M) = \sigma\tau(N)$
\end{lemma}

We can finally give the proof of the main theorem.
\subsubsection*{Proof of Theorem~\ref{thm:simu}.}
\emph{If $M \CAUmsrred N$, then $\sigma\tau(M) \CAUmrred \sigma\tau(N)$}.

\noindent We rewrite the hypothesis as
\[
M \Brred \steeq \Brred \steeq \cdots N
\]
Then the proof is a diagram chase based on Lemma~\ref{lem:simuBred} and Lemma~\ref{lem:simuStred}.

\section{Proofs about the abstract machine}
This section is devoted to the correctness proof for the abstract machine presented in Section~\ref{SECD}.
Before proceeding to the main theorem, we need some auxiliary definitions and lemmas.

\begin{definition}
A machine state $(\pi,D)$ is valid if $\jtm{(\pi,D)}$ and:
\begin{itemize}
\item for all $(q|\pure{M}|e)$ in $\pi$, $\pure{M}[e]$ is closed and all values in $e$ are also closed
\item all values in $D$ are closed
\end{itemize}
\end{definition}

\begin{lemma}\label{lem:valid}
If $\varsigma$ is valid and $\varsigma \mapsto \varsigma'$, then $\varsigma'$ is also valid.
\end{lemma}

\begin{corollary}
All reachable states are valid.
\end{corollary}

\begin{lemma}\label{lem:lookup}
Suppose $n[e]$ is closed: then $\erase{n[e]} \CAUmsrred e(n)$.
\end{lemma}

\begin{lemma}\label{lem:jctxval}
If $\jctx{\pi,D}$, then for all $V$ we have $\jtm{\pi,V::D}$ and $\cT(\pi,V::D) = \dec{\pi,D}[V]$.
\end{lemma}

\begin{lemma}\label{lem:auxinsp}
If $\cI((q;\qtr(\vect{q_i \derives C_i})),\pi,D)$ is defined, then $\dec{(\pi,D)}[q \derives \tinsp(\vect{q_i \derives C_i})] = \cE_\sigma[\bang_{q^*} \cF[q \derives \tinsp(\vect{q_i \derives C_i})]$
and $\cI((q;\qtr(\vect{q_i \derives C_i})),\pi,D) = \sigma\tau(q^*;\trailify{\cF[q \derives \tinsp(\vect{q_i \derives C_i})]})$
for some $\cE_\sigma$, $\cF$, $q^*$.
\end{lemma}
\begin{proof}
By induction on the derivation of $\jctx{\pi,D}$.
\end{proof}

\begin{lemma}\label{lem:admissible}
The following rules are admissible:
\begin{enumerate}
\item $(q \derives \erase{(\lambda\pure{M}[e])}) \; (q' \derives C) \CAUmsrred \qapp(q,q');\qba \derives \erase{\pure{M}[(\qrefl \derives C) \cdot e]}$
\item $\tlet(q \derives \bang~V, \erase{\pure{N}[1\cdot(e \circ \uparrow)]}) \CAUmsrred \qlet(q,\qrefl);\qbb;\trailify{\erase{\pure{N}[1 \cdot (e \circ \uparrow)]}[V]} \derives \erase{\pure{N}[V\cdot e]}$
\item $q \derives \tinsp(\vect{q_i \derives C_i}) \CAUmsrred q;\qtr(\vect{q_i});\qti \derives \erase{\cI((q;\qtr(\vect{q_i \derives C_i})),\pi,D)[\vect{\qrefl \derives C_i}]}$
\end{enumerate}
\end{lemma}
\begin{proof}
By induction, also using Lemma~\ref{lem:auxinsp} for the third rule.
\end{proof}

We can now prove that the abstract machine is correct.

\subsubsection*{Proof of Theorem \ref{thm:SECDsound}.}
\emph{For all valid $\varsigma$, $\varsigma \mapsto \varsigma'$ implies $\tmdec{\varsigma} \CAUmsrred \tmdec{\varsigma'}$.}

We proceed by cases on the transition $\varsigma \mapsto \varsigma'$, knowing that $\varsigma$ is valid by hypothesis and $\varsigma'$ by Lemma~\ref{lem:valid}. In rules 3, 7, 10 we also use Lemma~\ref{lem:jctxval} to infer the shape of the denotation of the target state. We need to prove the following statements:
\begin{enumerate}
\item $\dec{(\pi,D)}[q \derives \erase{(\pure{M}~\pure{N})[e]}] \CAUmsrred
  \dec{(\pi,D)}[q \derives ((\qrefl \derives \erase{\pure{M}[e]})~(\qrefl \derives \erase{\pure{N}[e]}))]$
\item $\dec{(\pi,D)}[q \derives ((q'' \derives \erase{(\lambda\pure{M})[e]})~(q' \derives C))] \CAUmsrred \\
  \dec{(\pi,D)}[q;\qapp(q'',q');\qba \derives \erase{\pure{M}[(\qrefl \derives C) \cdot e]}]$
\item $\dec{(\pi,D)}[q \derives \erase{(\lambda\pure{M})[e]}] \CAUmsrred
  \dec{(\pi,D)}[q \derives \erase{(\lambda\pure{M})[e]}]$
\item $\dec{(\pi,D)}[q \derives \erase{\tlet(M,N)[e]}] \CAUmsrred \\
  \dec{(\pi,D)}[q \derives \tlet(\qrefl \derives \erase{\pure{M}[e]}, \erase{\pure{N}[1 \cdot (e \circ \uparrow)]})]$
\item $\dec{(\pi,D)}[q \derives \tlet(q'\derives \bang V, \erase{\pure{N}[1 \cdot (e \circ \uparrow)]})] \CAUmsrred 
  \\
  \dec{(\pi,D)}[q;\qlet(q',\qrefl);\qbb;\trailify{\erase{\pure{N}[1 \cdot (e \circ \uparrow)]}[V]} \derives \erase{\pure{N}[V\cdot e]}]$
\item $\dec{(\pi,D)}[q \derives \erase{(\bang_{q'} \pure{M})[e]}] \CAUmsrred
  \dec{(\pi,D)}[q \derives \bang_{q';\trailify{\pure{M}[e]}} \erase{\pure{M}[e]}]$
\item $\dec{(\pi,D)}[q \derives \bang V] \CAUmsrred 
  \dec{(\pi,D)}[q \derives \bang V]$
\item $\dec{(\pi,D)}[q \derives \erase{\tinsp(\vect{M_9})[e]}] \CAUmsrred 
  \dec{(\pi,D)}[q \derives \tinsp(\qrefl \derives \erase{M_1[e]},\ldots,\qrefl \derives \erase{M_9[e]})]$
\item $\dec{(\pi,D)}[q \derives \tinsp(\vect{(q_i \derives C_i)}_{i=1,\ldots,9})] \CAUmsrred \\
  \dec{(\pi,D)}[q;\qtr(\vect{q_i});\qti \derives \erase{\cI((q;\qtr(\vect{q_i}),\pi,D)[\vect{\qrefl \derives C_i}]}]$
\item $\dec{(\pi,D)}[q \derives \erase{n[e]}] \CAUmsrred 
  \dec{(\pi,D)}[q \derives e(n)]$
\end{enumerate}
The statements number 2, 5, and 9 follow from Lemma~\ref{lem:admissible}; statement 10 is implied by Lemma~\ref{lem:lookup}; the other statements follow immediately from the definition of $\sred$ and $\tred$.

\end{document}